\documentclass[conference]{IEEEtran}
\IEEEoverridecommandlockouts
\usepackage{cite}
\usepackage{amsmath,amssymb,amsfonts}
\usepackage{algorithmic}
\usepackage{graphicx}
\usepackage{textcomp}
\usepackage{xcolor}
\def\BibTeX{{\rm B\kern-.05em{\sc i\kern-.025em b}\kern-.08em
    T\kern-.1667em\lower.7ex\hbox{E}\kern-.125emX}}

\usepackage{amsmath,amssymb,amsfonts,bm}
\usepackage{stmaryrd}
\usepackage[ruled,lined,french]{algorithm2e}
\usepackage{hyperref}

\usepackage{pgf}
\usepackage{subfig}
\graphicspath{{fig/}}

\usepackage{amsthm}
\newtheorem{theorem}{Theorem}
\usepackage{makecell}
\usepackage{booktabs} 

\newcommand{\intEnt}[2]{\left\{#1, \dots, #2\right\}}
\newcommand{\set}[1]{\left\{#1\right\}}

\newcommand{\myvector}[1]{{\bm{#1}}}
\newcommand{\mymatrix}[1]{{\bm{#1}}}
\newcommand{\eqdef}{\triangleq}

\newcommand{\diag}{\mathrm{diag}}

\newcommand{\suchthat}{\ | \ }

\newcommand{\E}{\mathrm{E}} 

\newcommand{\cent}{{\mathrm{c}}}

\newcommand{\R}{\mathbb{R}} 
\newcommand{\N}{\mathbb{N}} 

\newcommand{\rA}{\mathcal{A}}

\newcommand{\rK}{\mathcal{K}}
\newcommand{\rM}{\mathcal{M}}
\newcommand{\rN}{\mathcal{N}}

\newcommand{\ie}{{i.e.},}
\newcommand{\eg}{{e.g.},}

\begin{document}

\title{Split Covariance Intersection with Correlated Components for Distributed Estimation}

\author{\IEEEauthorblockN{Colin Cros}
\IEEEauthorblockA{\textit{Univ. Grenoble Alpes, CNRS, GIPSA-Lab} \\
Grenoble, France \\
colin.cros@gipsa-lab.fr}
\and
\IEEEauthorblockN{Pierre-Olivier Amblard}
\IEEEauthorblockA{\textit{Univ. Grenoble Alpes, CNRS, GIPSA-Lab}\\
Grenoble, France \\
pierre-olivier.amblard@cnrs.fr}
\and
\IEEEauthorblockN{Christophe Prieur}
\IEEEauthorblockA{\textit{Univ. Grenoble Alpes, CNRS, GIPSA-Lab} \\
Grenoble, France \\
christophe.prieur@gipsa-lab.fr}
\and
\IEEEauthorblockN{Jean-François Da Rocha}
\IEEEauthorblockA{\textit{Telespazio France} \\
Toulouse, France \\
jeanfrancois.darocha@telespazio.com}
}


\maketitle

\begin{abstract}
	This paper introduces a new conservative fusion method to exploit the correlated components within the estimation errors. Fusion is the process of combining multiple estimates of a given state to produce a new estimate with a smaller MSE. To perform the optimal linear fusion, the (centralized) covariance associated with the errors of all estimates is required. If it is partially unknown, the optimal fusion cannot be computed. Instead, a solution is to perform a conservative fusion. A conservative fusion provides a gain and a bound on the resulting MSE matrix which guarantees that the error is not underestimated. A well-known conservative fusion is the Covariance Intersection fusion. It has been modified to exploit the uncorrelated components within the errors. In this paper, it is further extended to exploit the correlated components as well. The resulting fusion is integrated into standard distributed algorithms where it allows exploiting the process noise observed by all agents. The improvement is confirmed by simulations.
\end{abstract}


\section{Introduction}

Distributed estimation is a recurrent problem in sensor networks. It consists in estimating the state of a dynamical system using a network of agents, \ie{} nodes equipped with sensors and with communication capabilities. Each agent is performing independent measurements of the state, and shares its estimate with its neighbors in the network. The fusion of the estimates received by an agent is a complex task, especially when the communication between the agents are limited. To optimally fuse several estimates, \ie{} with the smallest resulting Mean Square Error (MSE) matrix, the agent needs to know the covariances of the errors of each estimate and their cross-covariances. For example, the optimal fusion of two estimates is given by the Bar-Shalom-Campo's formulas \cite{bar1986effect}. However, in cooperating networks, each agent has only local knowledge; it can estimate the covariance of its own error but cannot estimate the cross-covariances with its neighbors' errors, as this would require knowledge of the whole network topology. Without these cross-covariances, the agent cannot calculate the MSE matrix of the fused estimate and must use \emph{conservative} bounds. A conservative fusion provides a bound on the MSE matrix which guarantees that the estimation error is not underestimated. The first conservative fusion proposed was the Covariance Intersection fusion (CI) \cite{julier1997nondivergent}. CI provides conservative bounds by considering that the estimation errors may be correlated to any degree. Generally, this assumption is very loose and tighter fusions have been derived when refined assumptions can be made. When the errors contain independent components, the Split CI (SCI) fusion provides tighter bounds \cite{julier2001general}. Furthermore, if the errors are partitioned into two components with only the cross-covariances between the first components unknown, another fusion rule called Partitioned CI (PCI) was proposed in \cite{petersen2011partitioned} and improved in \cite{ajgl2019rectification}. Another extension is the Inverse CI (ICI) which considers that the estimates were all obtained by combining some initial estimate with independent estimates \cite{noack2017decentralized, ajgl2020inverse}. All these improvements of the CI fusion use the structure of the errors to reduce the set of admissible cross-covariances and tighten the bounds. They have been applied to wide range of problems: \eg{} SLAM \cite{julier2007using}, cooperative localization \cite{li2013cooperative}, or cooperative perception \cite{lima2021data}.

In distributed estimation, the fusion of the estimates can be performed using CI. However, two elements can be used to produce tighter bounds. First, the independent measurements induce independent components in the errors. They can be exploited by the SCI fusion, as in \cite{julier2001general}. If the measurements are transmitted to the neighbors, the agents can optimally fuse them with their estimate, this method is known as Diffusion Kalman Filtering (DKF) with CI \cite{cattivelli2010diffusion, hu2011diffusion}. The second element is the process noise which is observed by all the agents. It reduces as well the set of admissible covariances: \eg{} the errors cannot be perfectly negatively correlated. The current fusions do not take advantage from such common noises as they do not consider correlated components.

In this paper, an extension of the SCI fusion, called Extended SCI (ESCI), is introduced. It is motivated by the distributed estimation problems in which it can exploit both the uncorrelated components of the errors (induced by the measurements) and their correlated components (induced by the process noise). It is integrated into standard algorithms and applied to an example inspired by Search-And-Rescue (SAR) missions.

The rest of the paper is organized as follows. Section~\ref{sec: Background} recalls the definitions of a conservative fusion and of the SCI fusion. Then, Section~\ref{sec: Extended SCI} presents the new ESCI fusion. Section~\ref{sec: Algorithms} details its integration into standard distributed estimation algorithms. The application is presented in Section~\ref{sec: Simulations}. Finally, Section~\ref{sec: Conclusion} concludes the paper.

\medbreak
\noindent \textbf{Notation.}
In the sequel, vectors are denoted in lowercase boldface letters \eg{} $\myvector{x} \in \R^n$, and matrices in uppercase boldface variables \eg{} $\mymatrix{M} \in \R^{n\times n}$. The notation $\E[\cdot]$ denotes the expected value. 
For two matrices $\mymatrix{A}$ and $\mymatrix{B}$, the notation $\mymatrix{A} \preceq \mymatrix{B}$  means that the difference $\mymatrix{B} - \mymatrix{A}$ is positive semi-definite. The unit simplex of $\R^n$ is denoted as $\rK^n \eqdef \set{\myvector{x} \in \R^n \suchthat \forall i, \, x_i \ge 0, \, \myvector{x}^{\intercal}\myvector{1} = 1}$.

\section{Background}\label{sec: Background}

Consider $N$ unbiased estimates $\myvector{\hat x}_i$ for $i \in \intEnt{1}{N}$ of a random variable $\myvector{x} \in \R^d$. The estimation errors are denoted as $\myvector{\tilde x}_i \eqdef \myvector{\hat x}_i - \myvector{x}$ and their covariances as $\mymatrix{\tilde P}_i \eqdef \E\left[\myvector{\tilde x}_i \myvector{\tilde x}_i^{\intercal}\right]$. A linear fusion is defined by a matrix of gains $\mymatrix{K} = \begin{bmatrix}\mymatrix{K}_1 & \dots & \mymatrix{K}_N\end{bmatrix} \in \R^{d \times Nd}$, with $\mymatrix{K}_i \in \R^{d\times d}$ and $\sum_i \mymatrix{K}_i = \mymatrix{I}_d$, as:
\begin{equation}
	\myvector{\hat x}_F(\mymatrix{K}) \eqdef \sum_{i=1}^N \mymatrix{K}_i \myvector{\hat x}_i = \mymatrix{K} \myvector{\hat x}_{\cent},
\end{equation}
where $\myvector{\hat x}_{\cent} \eqdef \begin{pmatrix}\myvector{\hat x}_1^{\intercal} & \dots & \myvector{\hat x}_N^{\intercal}\end{pmatrix}^{\intercal} \in \R^{Nd}$. The covariance of the error of the fused estimate depends on the gain $\mymatrix{K}$ and on the covariance of the error of $\myvector{\hat x}_{\cent}$, $\mymatrix{\tilde P}_{\cent}$:
\begin{equation}
	\mymatrix{\tilde P}_F(\mymatrix{K}, \mymatrix{\tilde P}_{\cent}) = \mymatrix{K} \mymatrix{\tilde P}_{\cent} \mymatrix{K}^{\intercal}.
\end{equation}
If $\mymatrix{\tilde P}_{\cent}$ is not entirely known but is only assumed to belong to some subset of admissible covariance matrices $\rA$, then $\mymatrix{\tilde P}_F(\mymatrix{K},\mymatrix{\tilde P}_{\cent})$ cannot be computed. In this case, an alternative is to provide a \emph{conservative} bound. A couple $(\mymatrix{K}, \mymatrix{B}_F)$ is said to generate a conservative fusion for the set $\rA$ if:
\begin{align}
	\forall \mymatrix{P}_{\cent} &\in \rA, & \mymatrix{\tilde P}_F(\mymatrix{K}, \mymatrix{P}_{\cent}) &\preceq \mymatrix{B}_F.
\end{align}
In other words, fusing the estimates with the gain $\mymatrix{K}$ ensures that the covariance of the error is bounded by $\mymatrix{B}_F$.

CI considers that the covariances of the errors $\mymatrix{\tilde P}_i$ are known but not their cross-covariances $\mymatrix{\tilde P}_{i,j}\eqdef \E\left[\myvector{\tilde x}_i \myvector{\tilde x}_j^{\intercal}\right]$. For any $\myvector{\omega} \in \rK^N$, CI provides a conservative fusion defined as:
\begin{align}\label{eq: CI fusion}
	\myvector{\hat x}_F &= \mymatrix{B}_F \sum_{i=1}^N \omega_i \mymatrix{\tilde P}_i^{-1}\myvector{\hat x}_i, &
	\mymatrix{B}_F^{-1} &= \sum_{i=1}^N \omega_i \mymatrix{\tilde P}_i^{-1}.
\end{align}
CI considers that the errors may be completely correlated. In distributed estimation, the estimators integrate independent measurements $\myvector{z}_i$ and have the following structure:
$$\myvector{\hat x}_i = (\mymatrix{I} - \mymatrix{K}\mymatrix{H}_i) \myvector{\hat x}_i^- + \mymatrix{K} \myvector{z}_i.$$
Therefore, the errors $\myvector{\tilde x}_i$ cannot be perfectly correlated and SCI produces tighter bounds. It considers that the estimation errors are split into a correlated and an uncorrelated component as:
\begin{equation}\label{eq: SCI error decomposition}
	\myvector{\tilde x}_i = \myvector{\tilde x}_i^{(1)} + \myvector{\tilde x}_i^{(2)},
\end{equation}
where the components $\myvector{\tilde x}_i^{(1)}$ are correlated to an unknown degree while the components $\myvector{\tilde x}_i^{(2)}$ are uncorrelated between each other and with the $\myvector{\tilde x}_i^{(1)}$. The covariances of $\myvector{\tilde x}_i^{(1)}$ and $\myvector{\tilde x}_i^{(2)}$ are denoted as $\mymatrix{\tilde P}_i^{(1)}$ and $\mymatrix{\tilde P}_i^{(2)}$ (and are assumed known). For any $\myvector{\omega} \in \rK^N$, SCI provides a conservative fusion defined as:
\begin{subequations}\label{eq: SCI fusion}
	\begin{align}
		\myvector{\hat x}_F &= \mymatrix{B}_F \sum_{i=1}^N \omega_i \left(\mymatrix{\tilde P}_i^{(1)}+ \omega_i\mymatrix{\tilde P}_i^{(2)}\right)^{-1}\myvector{\hat x}_i,\\
		\mymatrix{B}_F^{-1} &= \sum_{i=1}^N \omega_i \left(\mymatrix{\tilde P}_i^{(1)}+ \omega_i\mymatrix{\tilde P}_i^{(2)}\right)^{-1}.
	\end{align}
\end{subequations}
The parameter $\myvector{\omega}$ must be chosen: optimized or empirically tuned with \eg{} the methods in \cite{niehsen2002information} or \cite{franken2005improved}.

\section{Extended SCI}\label{sec: Extended SCI}

\subsection{Motivation: Limits of the SCI fusion}

The CI fusion considers that the errors can be correlated to any degree. In distributed estimation problems, the estimates incorporate independent measurements, and therefore their errors contain independent components. The SCI fusion has been defined to exploit theses independent terms.

Moreover, the state to estimate is often disturbed by an additive process noise $\mymatrix{w}$. This noise is added to all the estimation errors during the prediction step. All the errors share a common component, which also reduces the space of admissible centralized covariance matrices $\rA$. For example, the errors cannot be perfectly negatively correlated. However, SCI cannot exploit this correlated component as it can only handle uncorrelated components. The new ESCI fusion is defined to overcome this limitation.

\subsection{Definition of the ESCI fusion}

Consider that the estimation errors are split into two components as in \eqref{eq: SCI error decomposition}. The first components $\myvector{\tilde x}_i^{(1)}$ are still correlated to an unknown degree. The second components $\myvector{\tilde x}_i^{(2)}$ are not assumed uncorrelated, but are assumed to have known second moments. Introduce the centralized errors,
\begin{align*}
	\myvector{\tilde x}_{\cent}^{(l)} &\eqdef \begin{pmatrix}
		\myvector{\tilde x}_1^{(l)\intercal} & \dots & \myvector{\tilde x}_N^{(l)\intercal}
	\end{pmatrix}^{\intercal}, & l &\in \set{1,2},
\end{align*}
whose covariances and cross-covariances are denoted as $\mymatrix{\tilde P}_{\cent}^{(l)} \eqdef \E\left[\myvector{\tilde x}_{\cent}^{(l)}\myvector{\tilde x}_{\cent}^{(l)\intercal}\right]$ and $\mymatrix{\tilde P}_{\cent}^{(1,2)} \eqdef \E\left[\myvector{\tilde x}_{\cent}^{(1)}\myvector{\tilde x}_{\cent}^{(2)\intercal}\right]$. The matrices $\mymatrix{\tilde P}_{\cent}^{(2)}$ and $\mymatrix{\tilde P}_{\cent}^{(1,2)}$ are known, but only the diagonal blocks of $\mymatrix{\tilde P}_{\cent}^{(1)}$ (corresponding to the covariances $\mymatrix{\tilde P}_i^{(1)}$) are known.
If $\mymatrix{\tilde P}_{\cent}^{(1,2)} \ne \mymatrix{0}$, the errors \eqref{eq: SCI error decomposition} are virtually re-splittable to set $\mymatrix{\tilde P}_{\cent}^{(1,2)} = \mymatrix{0}$ by letting:
\begin{subequations}
	\begin{align}
		\myvector{\tilde x}_{\cent}^{(1)} &\gets \myvector{\tilde x}_{\cent}^{(1)} - \mymatrix{\tilde P}_{\cent}^{(1,2)}(\mymatrix{\tilde P}_{\cent}^{(2)})^{-1} \myvector{\tilde x}_{\cent}^{(2)}, \\ 
		\myvector{\tilde x}_{\cent}^{(2)} &\gets \myvector{\tilde x}_{\cent}^{(2)} + \mymatrix{\tilde P}_{\cent}^{(1,2)}(\mymatrix{\tilde P}_{\cent}^{(2)})^{-1} \myvector{\tilde x}_{\cent}^{(2)}.
	\end{align}
\end{subequations}
The errors $\myvector{\tilde x}_{\cent}^{(1)}$ and $\myvector{\tilde x}_{\cent}^{(2)}$ satisfy the same properties: only the off-diagonal blocks of $\mymatrix{\tilde P}_{\cent}^{(1)}$ are unknown.
We therefore assume without loss of generality that $\mymatrix{\tilde P}_{\cent}^{(1,2)} = \mymatrix{0}$. In this splitting, $\myvector{\tilde x}_i^{(2)}$ contains all \emph{known} components. In distributed estimation, it will contain the independent measurement noise plus the common process noise as illustrated in the next section.

For any $\myvector{\omega} \in \rK^N$, the ESCI fusion is defined as:
\begin{subequations}\label{eq: ESCI fusion}
	\begin{align}
		\myvector{\hat x}_F &= \mymatrix{B}_F \mymatrix{B}_{\cent}^{-1} \mymatrix{H}\myvector{\hat x}_{\cent}, &
		\mymatrix{B}_F^{-1} &= \mymatrix{H}^{\intercal} \mymatrix{B}_{\cent}^{-1} \mymatrix{H},
	\end{align}
	with:
	\begin{align}
		\mymatrix{H} &= \myvector{1}_N \otimes \mymatrix{I}_d,\\
		\mymatrix{B}_{\cent}^{(1)} &= \diag\left(\frac{1}{\omega_1}\mymatrix{\tilde P}_1^{(1)}, \dots, \frac{1}{\omega_N}\mymatrix{\tilde P}_N^{(1)}\right), \\
		\mymatrix{B}_{\cent} &= \mymatrix{B}_{\cent}^{(1)} + \mymatrix{\tilde P}_{\cent}^{(2)}.
	\end{align}
\end{subequations}
The ESCI is a generalization of the SCI in the sense that if $\mymatrix{P}_{\cent}^{(2)}$ is block diagonal, then \eqref{eq: ESCI fusion} and \eqref{eq: SCI fusion} define the same fusion.

\begin{theorem}
	For any $\myvector{\omega}\in \rK^N$, the ESCI fusion defined in \eqref{eq: ESCI fusion} is conservative for the set:
	\begin{multline}
		\rA_\mathrm{ESCI} \eqdef \left\{\mymatrix{P}_{\cent}^{(1)} + \mymatrix{\tilde P}_{\cent}^{(2)} \suchthat \mymatrix{P}_{\cent}^{(1)} \succeq \mymatrix{0} \text{ and } \right.\\
		\left. \forall i \in \intEnt{1}{N}, \, \mymatrix{P}_{i}^{(1)} = \mymatrix{\tilde P}_i^{(1)} \right\}.
	\end{multline}
\end{theorem}
\begin{proof}
	The proof is the same as for the SCI fusion \cite{julier2001general}. For any $\myvector{\omega} \in \rK^N$, the matrix $\mymatrix{B}_{\cent}^{(1)}$ is a conservative bound on the centralized covariance of $\myvector{\tilde x}_{\cent}^{(1)}$ \cite{julier2001general}. Therefore, $\mymatrix{B}_{\cent}$ is a conservative bound for the centralized covariance of $\myvector{\tilde x}_{\cent}$. The fusion is then obtained by applying the gain:
	$\mymatrix{K} = \left(\mymatrix{H}^{\intercal} \mymatrix{B}_{\cent}^{-1} \mymatrix{H}\right)^{-1} \mymatrix{H}^{\intercal} \mymatrix{B}_{\cent}^{-1}$.
\end{proof}

\subsection{Special case of a common noise}

When the correlated components of the errors come from a common noise $\myvector{w}$, the ESCI expressions can be simplified. Consider that the estimation errors are split as :
\begin{equation}\label{eq: ESCI error decomposition}
	\myvector{\tilde x}_i = \mymatrix{\tilde x}_i^{(1)} + \myvector{\tilde x}_i^{\mathrm{(ind)}} + \mymatrix{M}_i \myvector{w},
\end{equation}
where the components $\mymatrix{\tilde x}_i^{(1)}$ are correlated to an unknown degree, the components $\myvector{\tilde x}_i^{\mathrm{(ind)}}$ are uncorrelated between each other, with the $\myvector{\tilde x}_i^{(1)}$ and with $\myvector{w}$, the matrices $\mymatrix{M}_i$ are known, and $\myvector{w}$ is a common independent noise. The covariances of each component are known and denoted as $\mymatrix{P}_i^{\mathrm{(ind)}}$ for $\myvector{\tilde x}_i^{\mathrm{(ind)}}$ and $\mymatrix{Q}$ for $\myvector{w}$. In this case the fusion \eqref{eq: ESCI fusion} becomes:
\begin{subequations}\label{eq: ESCI fusion with common noise}
	\begin{align}
		\myvector{\hat x}_F &= \mymatrix{B}_F\sum_{i=1}^N \omega_i(\mymatrix{I}_d - \mymatrix{S}_1\mymatrix{S}_0^{-1} \mymatrix{M}_i^{\intercal}) \mymatrix{\tilde P}_i^{\prime{-1}}\myvector{\hat x}_i, \\
		\mymatrix{B}_F^{-1} &= \sum_{i=1}^N \omega_i \mymatrix{\tilde P}_i^{\prime{-1}} - \mymatrix{S}_1 \mymatrix{S}_0^{-1} \mymatrix{S}_1^{\intercal},
	\end{align}
	with $\mymatrix{\tilde P}_i' \eqdef \mymatrix{\tilde P}_i^{(1)} + \omega_i \mymatrix{\tilde P}_i^{\mathrm{(ind)}}$ and :
	\begin{align}
		\mymatrix{S}_0 &= \sum_{i=1}^N \omega_i \mymatrix{M}_i^{\intercal}  \mymatrix{\tilde P}_i^{\prime{-1}} \mymatrix{M}_i + \mymatrix{Q}^{-1},\\
		\mymatrix{S}_1 &= \sum_{i=1}^N \omega_i  \mymatrix{\tilde P}_i^{\prime{-1}} \mymatrix{M}_i.
	\end{align}
\end{subequations}
The advantage of \eqref{eq: ESCI fusion with common noise} over \eqref{eq: ESCI fusion} is that  \eqref{eq: ESCI fusion with common noise} requires to invert $N+1$ matrices of size $d$ while $\eqref{eq: ESCI fusion}$ requires the inversion of one matrix of size $Nd$. As the cost of an inversion of a matrix of size $n$ is a $O(n^3)$, \eqref{eq: ESCI fusion with common noise} is more efficient.

If all the matrices $\mymatrix{M}_i = \mymatrix{I}_d$, then \eqref{eq: ESCI fusion with common noise} simplifies further to:
\begin{align}\label{eq: ESCI fusion with common additive noise}
	\myvector{\hat x}_F &=\mymatrix{B}_0 \sum_{i=1}^N \omega_i \mymatrix{\tilde P}_i^{\prime-1} \myvector{\hat x}_i, &
	\mymatrix{B}_F &= \mymatrix{B}_0 + \mymatrix{Q},
\end{align}
where $\mymatrix{B}_0^{-1} = \sum_{i=1}^N \omega_i\mymatrix{\tilde P}_i^{\prime-1}$.
This case is equivalent to first fuse the uncorrupted estimates $\myvector{\tilde x}_i^{(1)} + \myvector{\tilde x}_i^{\mathrm{(ind)}}$ using SCI and then add the noise.

\begin{figure}
	\centering
	\null\hfill
	\subfloat[CI fusion.]{\input{fig/fusion_ci.tex}}
	\hfill
	\subfloat[SCI fusion.]{\input{fig/fusion_sci.tex}}
	\hfill
	\subfloat[ESCI fusion.]{\input{fig/fusion_gsci.tex}}
	\hfill\null
	\caption{Comparison of the bounds provided by CI, SCI and ESCI. The dotted ellipses represent the covariances $\mymatrix{\tilde P}_1$ and $\mymatrix{\tilde P}_2$, the grey ellipses are the bound obtained with $\myvector{\omega} = \begin{pmatrix} 2k/10 & 1-2k/10 \end{pmatrix}^{\intercal}$ and $k \in \intEnt{0}{5}$, and the dark ellipse is the bound that minimizes the trace. The numerical values used are: $\mymatrix{\tilde P}_1^{(1)} = [[1, -2], [-2, 5]]$, $\mymatrix{\tilde P}_2^{(1)} = [[9, -1], [-1, 1]]$, $\mymatrix{\tilde P}_1^{\mathrm{(ind)}} = [[2, 0], [0, 9]]$, $\mymatrix{\tilde P}_2^{\mathrm{(ind)}} = [[9, 3], [3, 2]]$, $\mymatrix{Q} = [[2, 2], [2, 2]]$, and $\mymatrix{M}_1 = \mymatrix{M}_2 = \mymatrix{I}$.}
	\label{fig: Comparison of the fusions}
\end{figure}
To illustrate the interest of the ESCI fusion, consider the fusion of two estimates whose errors are split according to \eqref{eq: ESCI error decomposition}. Theses estimates can be split using CI (without considering the splitting), using SCI (by grouping the correlated component $\mymatrix{M}_i \myvector{w}$ with the first component), or using ESCI. Figure~\ref{fig: Comparison of the fusions} compares the bounds obtained with the three fusions. It can be observed that the ESCI bounds are tighter, as expected.

\section{Distributed estimation algorithms}\label{sec: Algorithms}

This section presents the integration of the ESCI fusion into distributed estimation algorithms.

\subsection{System model}

Consider a system parameterized by a discrete-time state-space model. The state at time $k \in \N$ is denoted as $\myvector{x}(k) \in \R^d$. It is assumed to follow the following linear dynamics:
\begin{subequations}
	\begin{align}
		\myvector{x}(0) &\sim \rN(\myvector{x}_0, \mymatrix{\tilde P}_0), \\
		\forall k \in \N, \qquad \myvector{x}(k+1) &= \mymatrix{F}\myvector{x}(k) + \myvector{w}(k+1),
	\end{align}
	where $\myvector{x}_0$ is the initial state, $\mymatrix{\tilde P}_0$ the covariance matrix of the initial uncertainty, $\mymatrix{F}\in\R^{d \times d}$ the evolution matrix, and $\myvector{w}(k)$ the process noise at time $k$.
\end{subequations}
The system is estimated by a network of $N$ agents. The agents are equipped with sensors to observe the state, and have (limited) communication capabilities with their neighbors in the network. As the estimation algorithms are symmetrical between the agents, a focus is made on one particular agent indexed by $i \in \intEnt{1}{N}$. The set of neighbors of Agent~$i$ is denoted as $\rN_i$ and its neighborhood as $\rM_i \eqdef \rN_i \cup \set{i}$. At each time step $k \in \N$, Agent~$i$ performs the measurement $\myvector{z}_i(k) \in \R^{m_i}$:
\begin{equation}
	\myvector{z}_i(k) = \mymatrix{H}_{i} \myvector{x}(k) + \myvector{v}_i(k),
\end{equation}
where $\mymatrix{H}_i \in \R^{m_i \times d}$ is the observation matrix and $\myvector{v}_i(k)$ is the observation noise at time $k$.

The process noise and the observation noises of the agents are assumed: $(i)$ zero-mean, $(ii)$ white, and $(iii)$ uncorrelated between each other and with the error of the initial state. The covariance matrix of the process noise is denoted as $\mymatrix{Q}$ and that of the measurement noise of Agent~$i$ as $\mymatrix{R}_i$. 
In this model, all the matrices have been assumed time-independent for the sake of clarity. This work can however be adapted with time-varying matrices.

The estimate of Agent~$i$ at time $k$ based on the measurements up to time $l$ is denoted as $\myvector{\hat x}_i(k|l)$. The estimation error is $\myvector{\tilde x}_i(k|l) \eqdef \myvector{\hat x}_i(k|l) - \myvector{x}(k)$. Agent~$i$ also estimates a conservative bound, denoted as $\mymatrix{P}_i(k|l)$, of the covariance of its error. Here, conservative means that:
\begin{equation}
	\mymatrix{P}_i(k|l) - \E\left[\myvector{\tilde x}_i(k|l)\myvector{\tilde x}_i(k|l)^{\intercal}\right] \succeq \mymatrix{0}.
\end{equation}
At each time step, Agent~$i$ exchanges information with its neighbors $j \in \rN_i$. Three levels of communication are considered, labeled L1, L2 and L3. They restrict the amount of information that the agents can send (Level L1 being the more restrictive). The information transmitted in each level is given in Table~\ref{tab: Transmitted parameters} and is described in the next paragraph.

\subsection{Algorithm description}

The distributed estimation algorithm described below was proposed in \cite{julier2001general} for Levels L1 and L2. For Level L3, there are slight modifications, the resulting algorithm is the DKF algorithm proposed in \cite{hu2011diffusion}. These two algorithms are represented by the diagrams in Figure~\ref{fig: Schema distributed algorithms}. The estimation algorithm has four steps detailed below from the perspective of Agent~$i$.
\begin{figure}
	\centering
	\subfloat[Without measurement transmissions (Levels L1 and L2).]{\fontsize{6pt}{1em}{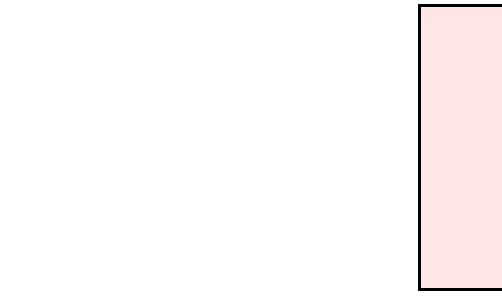}}\\
	\subfloat[With measurement transmissions (Level L3).]{\fontsize{6pt}{1em}{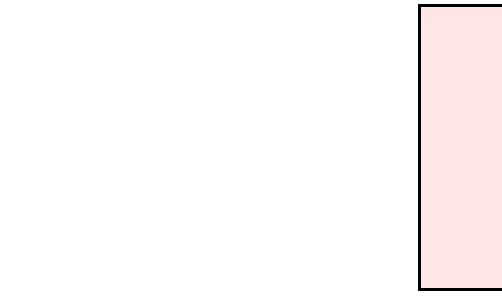}}
	\caption{Diagrams of the distributed algorithms.}
	\label{fig: Schema distributed algorithms}
\end{figure}
\begin{enumerate}
	\item \textbf{Prediction.} The state is predicted using the evolution model.
	\begin{subequations}
		\begin{align}
			\myvector{\hat x}_i(k|k-1) &= \mymatrix{F} \myvector{\hat x}_i(k-1|k-1), \\
			\mymatrix{P}_i(k|k-1) &= \mymatrix{F} \mymatrix{P}(k-1|k-1) \mymatrix{F}^{\intercal} + \mymatrix{Q}.
		\end{align}
	\end{subequations}
	\item \textbf{Update.} The prediction is updated using the measurement $\myvector{z}_i(k)$ (as in a Kalman filter).
	\begin{subequations}
		\begin{align}
			\myvector{\hat x}_i^\mathrm{(a)}(k|k) &= \mymatrix{P}_i^\mathrm{(a)}(k|k)\mymatrix{P}_i(k|k-1)^{-1}\myvector{\hat x}_i(k|k-1) \notag{}\\
			&\qquad+ \mymatrix{P}_i^\mathrm{(a)}(k|k)\mymatrix{H}_i^{\intercal} \mymatrix{R}_i^{-1} \myvector{z}_i(k), \\
			\mymatrix{P}_i^\mathrm{(a)}(k|k)^{-1} &= \mymatrix{P}_i(k|k-1)^{-1} + \mymatrix{H}_i^{\intercal} \mymatrix{R}_i^{-1} \mymatrix{H}_i.
		\end{align}
	\end{subequations}
	The superscript $\mathrm{(a)}$ (for autonomous) indicates that this estimate is obtained solely  from the measurements of Agent~$i$. Agent~$i$ sends this estimate to its neighbors, and receives their estimates. The complete list of parameters transmitted for each communication level is given in Table~\ref{tab: Transmitted parameters}.
	\item \textbf{Fusion.} The estimates received $\myvector{\hat x}_j$, $j\in\rN_i$, are fused with the prediction $\myvector{\hat x}_i(k|k-1)$. This fusion provides gains $\mymatrix{K}_j$ for $j\in\rM_i$ and a conservative bound $\mymatrix{B}_F$.
	\begin{subequations}
		\begin{align}
			\myvector{\hat x}_i^F(k|k-1) &= \mymatrix{K}_i \myvector{\hat x}_i(k|k-1) + \sum_{j\in\rN_i} \mymatrix{K}_j \myvector{\hat x}_j \\
			\mymatrix{P}_i^F(k|k-1) &= \mymatrix{B}_F.
		\end{align}
	\end{subequations}
	The details of the fusion are given in Paragraph~\ref{ssec: Fusions}.
	\item \textbf{Update.} The fused estimate is updated using the measurement $\myvector{z}_i(k)$ as in step~$2$.
	\begin{subequations}\label{eq: Final update}
		\begin{align}
			\myvector{\hat x}_i(k|k) &= \mymatrix{P}_i(k|k)\mymatrix{P}_i^F(k|k-1)^{-1}\myvector{\hat x}_i(k|k-1) \notag{}\\
			&\qquad+ \mymatrix{P}_i(k|k)\mymatrix{H}_i^{\intercal} \mymatrix{R}_i^{-1} \myvector{z}_i(k), \\
			\mymatrix{P}_i(k|k)^{-1} &= \mymatrix{P}_i^F(k|k-1)^{-1} + \mymatrix{H}_i^{\intercal} \mymatrix{R}_i^{-1} \mymatrix{H}_i.
		\end{align}
	\end{subequations}
\end{enumerate}
\begin{table}
	\centering
	\begin{tabular}{cl}
	\toprule
	Level & Parameters transmitted\\
	\midrule
	C1 & $\myvector{\hat x}_i^{(a)}(k|k)$, $\mymatrix{P}_i^{(a)}(k|k)$\\
	C2 & $\myvector{\hat x}_i^{(a)}(k|k)$, $\mymatrix{P}_i^{(a)}(k|k)$, $\mymatrix{H}_i^{\intercal} \mymatrix{R}_i^{-1} \mymatrix{H}_i$\\
	C3 & $\myvector{\hat x}_i(k|k-1)$, $\mymatrix{P}_i(k|k-1)$, $\mymatrix{H}_i^{\intercal} \mymatrix{R}_i^{-1}\myvector{z}_i(k)$, $\mymatrix{H}_i^{\intercal} \mymatrix{R}_i^{-1} \mymatrix{H}_i$\\
	\bottomrule
\end{tabular}
	\caption{Parameters transmitted by communication levels.}
	\label{tab: Transmitted parameters}
\end{table}
In Level L3, the measurements of the neighbors are optimally fused during the update of Step $4$. Therefore, there is no need to perform Step~$2$, and \eqref{eq: Final update} becomes:
\begin{subequations}
	\begin{align}
		\myvector{\hat x}_i(k|k) &= \mymatrix{P}_i(k|k)\mymatrix{P}_i^F(k|k-1)^{-1}\myvector{\hat x}_i(k|k-1) \notag{}\\
		&\qquad+ \mymatrix{P}_i(k|k)\sum_{j \in \rM_i }\mymatrix{H}_j^{\intercal} \mymatrix{R}_j^{-1} \myvector{z}_j(k), \\
		\mymatrix{P}_i(k|k)^{-1} &= \mymatrix{P}_i^F(k|k-1)^{-1} + \sum_{j \in \rM_i}\mymatrix{H}_j^{\intercal} \mymatrix{R}_j^{-1} \mymatrix{H}_j.
	\end{align}
\end{subequations}

\subsection{Fusions}\label{ssec: Fusions}

This paragraph details the fusion step of the distributed estimation algorithm. It is during this step that ESCI is used instead of the classical fusions to exploit the process noise.

\subsubsection{Level L1}
With Level L1, the fusion can only be performed using CI \eqref{eq: CI fusion}. Neither the independent components induced by the measurements, nor the correlated components induced by the process noise can be used.

\subsubsection{Level L2}
With Level L2, the authors of \cite{julier2001general} propose to fuse the estimates using SCI in order to exploit the independent component. Indeed, the error on the transmitted estimate is:
\begin{multline}
	\myvector{\tilde x}_j^\mathrm{(a)}(k|k) = \mymatrix{P}_j^\mathrm{(a)}(k|k)\mymatrix{P}_j(k|k-1)^{-1}\myvector{\tilde x}_j(k|k-1) \\
	+ \mymatrix{P}_j^\mathrm{(a)}(k|k)\mymatrix{H}_j^{\intercal} \mymatrix{R}_j^{-1} \myvector{v}_j(k).
\end{multline}
It is split as \eqref{eq: SCI error decomposition} with:
\begin{subequations}
	\begin{align}
		\myvector{\tilde x}_j^{(1)} &= \mymatrix{P}_j^\mathrm{(a)}(k|k)\mymatrix{P}_j(k|k-1)^{-1}\myvector{\tilde x}_j(k|k-1), \\
		\myvector{\tilde x}_j^{(2)} &= \mymatrix{P}_j^\mathrm{(a)}(k|k)\mymatrix{H}_j^{\intercal} \mymatrix{R}_j^{-1} \myvector{v}_j(k).
	\end{align}
\end{subequations}
The covariances of both components are computable from the parameters transmitted.
This decomposition takes advantage from the independent measurements but ignores the fact that the terms $\myvector{\tilde x}_j^{(1)}$ all contain the process noise $\myvector{w}(k)$:
\begin{equation}\label{eq: Error prediction}
	\myvector{\tilde x}_j(k|k-1) = \mymatrix{F}\myvector{\tilde x}_j(k-1|k-1) - \myvector{w}(k).
\end{equation}
The error \eqref{eq: Error prediction} can be rexpressed as \eqref{eq: ESCI error decomposition} with:
\begin{subequations}
	\begin{align}
		\myvector{\tilde x}_j^{(1)} &= \mymatrix{P}_j^\mathrm{(a)}(k|k)\mymatrix{P}_j(k|k-1)^{-1}\mymatrix{F}\myvector{\tilde x}_j(k-1|k-1), \\
		\myvector{\tilde x}_j^{\mathrm{(ind)}} &= \mymatrix{P}_j^\mathrm{(a)}(k|k)\mymatrix{H}_j^{\intercal} \mymatrix{R}_j^{-1} \myvector{v}_j(k), \\
		\mymatrix{M}_j &= -\mymatrix{P}_j^\mathrm{(a)}(k|k)\mymatrix{P}_j(k|k-1)^{-1}.
	\end{align}
	All the terms required for the ESCI fusion \eqref{eq: ESCI fusion with common noise} are also computable from the parameters transmitted.
\end{subequations}

\subsubsection{Level L3}
With Level L3, the authors of \cite{hu2011diffusion} propose to fuse the estimates using CI. However, as expressed in \eqref{eq: Error prediction}, the predicted estimates are all corrupted by the process noise. It is therefore more interesting to fuse them using ESCI \eqref{eq: ESCI fusion with common additive noise}. This fusion is equivalent to fuse the estimates $\myvector{\hat x}_j(k-1|k-1)$ using CI before performing the prediction step.

\section{Simulations}\label{sec: Simulations}

\begin{figure}
	\centering
	{\input{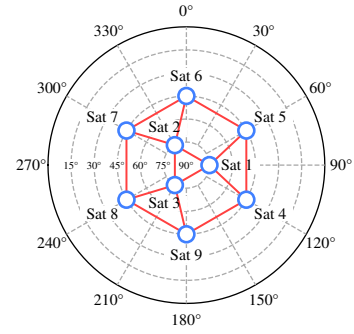}}
	\caption{Skyplot (azimuth / elevation) of the positions of the satellites. (A satellite in the center has an elevation of $90^{\circ}$, \ie{} is at zenith.) The lines represent the edges of the network.}
	\label{fig: Skyplot}
\end{figure}
To illustrate the interest of the ESCI fusion, the algorithms are applied in an example inspired by SAR missions. A distress signal is sent by an emitter of unknown position and is received by a network of satellites. The objective is to estimate the position $\myvector{p} \in \R^3$ of the emitter. To do so, the satellites measure the times of reception of the distress signal and deduce the \emph{pseudo-ranges} from the emitter. The pseudo-range is a biased version of the distance obtained by measuring the time-of-flight of the signal. The bias comes from the fact that the time of emission of the signal is unknown: the unknown clock-offet $\tau$ generates a bias $\beta \eqdef c \tau$ on the range measurement .
In this context, the state to estimate is the location of the emitter and the bias: $\myvector{x} = \begin{pmatrix}	\myvector{p}^{\intercal} & \beta \end{pmatrix}^{\intercal} \in \R^4$. This state is modeled with a slow dynamic, the evolution matrix is $\mymatrix{F} = \mymatrix{I}_4$ and the process noise has covariance $\mymatrix{Q} = \sigma_w^2 \mymatrix{I}_4$ with $\sigma_w = 5$~m.
The observation matrix of Satellite~$i$ is $\mymatrix{H}_i = \begin{bmatrix} \myvector{u}_i & 1 \end{bmatrix} \in \R^{1 \times 4}$ where $\myvector{u}_i$ is the unit vector pointing from the satellite to the emitter. The variance of the measurement noise is $R_i = \sigma_m^2$, with $\sigma_m = 10$~m. The network of satellites is represented with a skyplot in Figure~\ref{fig: Skyplot}.

The algorithms without and with measurement transmissions have been applied with the standard fusions and with the new ESCI fusion. For each fusion the parameter $\myvector{\omega}$ has been optimized to minimize the trace of the bound. Figure~\ref{fig: Simulation curves} presents the evolution of the estimated variance bounds and the MSEs computed over $10,000$ runs of $20$~iterations. We observe that for all fusions the MSEs are, as expected, lower than the bounds which confirms the conservativeness. For both algorithms, the ESCI fusions provide tighter bounds. For the algorithm without measurement transmissions, we observe reductions of the variance bounds with respect to the SCI fusion of about $19$\% for Sat.~$1$ and $11$\% for Sat.~$4$ in the horizontal plane (North / East) and of about  $23$\% for Sat.~$1$ and $18$\% for Sat.~$4$ in the vertical plane (Up). For the algorithm with measurement transmissions, the reductions of the variance bounds with respect to the CI fusion are of about $5$\% for both satellites in the horizontal plane (North / East) and of about $16$\% for Sat.~$1$ and $12$\% for Sat.~$4$ in the vertical plane (Up). This example confirms that exploiting the common process noise helps to improve the accuracy.
\begin{figure*}
	\centering
	\null\hfill
	\subfloat[Without measurement transmissions (Levels L1 and L2).]{\input{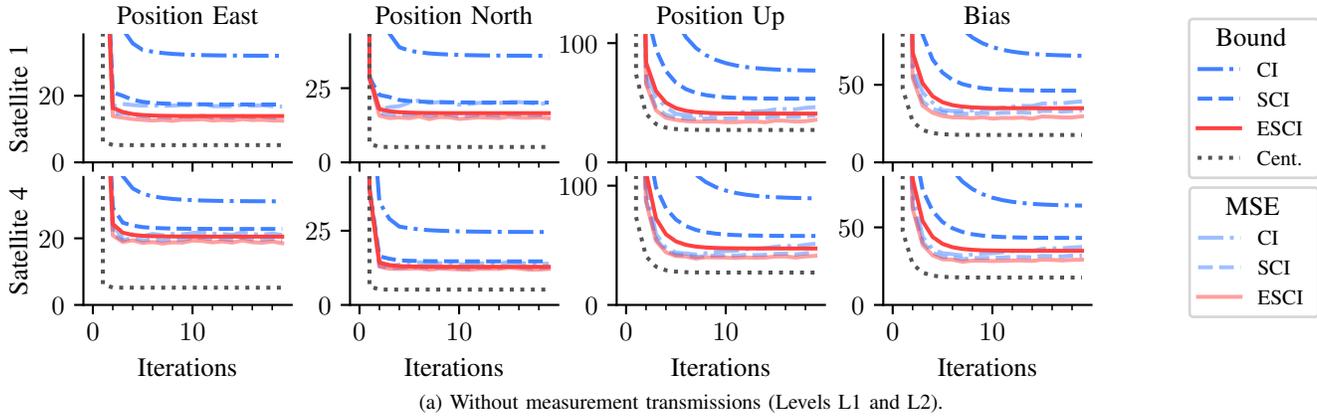}}
	\hfill
	\subfloat[With measurement transmissions (Level L3).]{\input{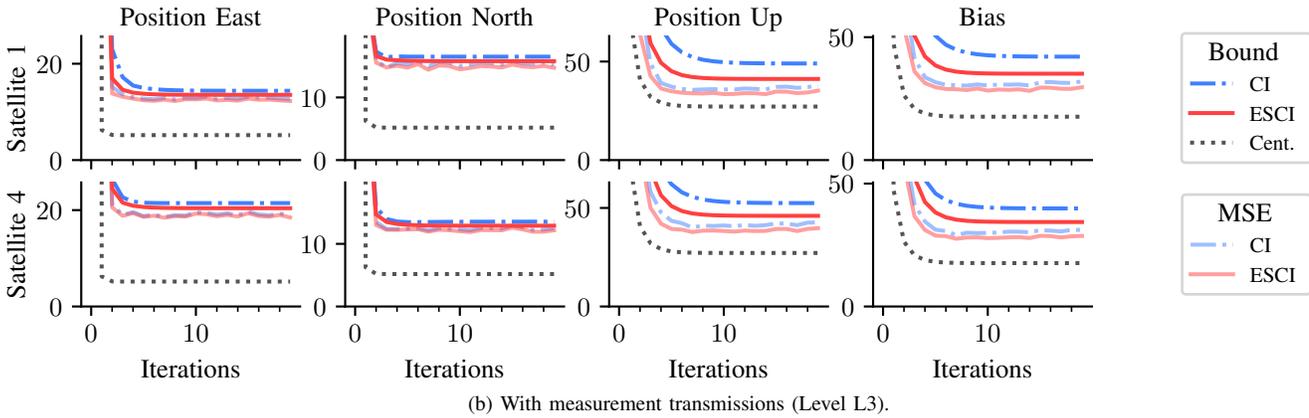}}
	\hfill\null
	\caption{Estimated variance bounds (matte curves) and MSEs (semi-transparent curves) for Satellite~$1$ and Satellite~$4$. The bound labeled '{Cent.}' were obtained with a centralized Kalman Filter to represent the optimal reachable performances.}
	\label{fig: Simulation curves}
\end{figure*}

\section{Conclusion}\label{sec: Conclusion}

This paper has introduced a new conservative fusion inspired by the SCI fusion. This fusion is designed to exploit all known components in the estimation errors, uncorrelated or not. When applied to distributed estimation problems, this fusion exploits the commonly observed process noise to produce tighter bounds. This fusion can be used with standard distributed algorithms, such as the DKF, with no additional communication requirements or complexity than the current fusions. The simulations have confirmed the interest of this fusion with significant bound reductions, particularly important for applications such as SAR.

\bibliographystyle{plain}
\bibliography{references}

\begin{thebibliography}{10}

\bibitem{ajgl2019rectification}
Ji{\v{r}}{\'\i} Ajgl and Ond{\v{r}}ej Straka.
\newblock Rectification of partitioned covariance intersection.
\newblock In {\em 2019 American Control Conference (ACC)}, pages 5786--5791.
  IEEE, 2019.

\bibitem{ajgl2020inverse}
Ji{\v{r}}{\'\i} Ajgl and Ond{\v{r}}ej Straka.
\newblock Inverse covariance intersection fusion of multiple estimates.
\newblock In {\em 2020 IEEE 23rd International Conference on Information Fusion
  (FUSION)}, pages 1--8. IEEE, 2020.

\bibitem{bar1986effect}
Yaakov Bar-Shalom and Leon Campo.
\newblock The effect of the common process noise on the two-sensor fused-track
  covariance.
\newblock {\em IEEE Transactions on aerospace and electronic systems},
  (6):803--805, 1986.

\bibitem{cattivelli2010diffusion}
Federico~S Cattivelli and Ali~H Sayed.
\newblock Diffusion strategies for distributed kalman filtering and smoothing.
\newblock {\em IEEE Transactions on automatic control}, 55(9):2069--2084, 2010.

\bibitem{franken2005improved}
Dietrich Franken and Andreas Hupper.
\newblock Improved fast covariance intersection for distributed data fusion.
\newblock In {\em 2005 7th International Conference on Information Fusion},
  volume~1, pages 7--pp. IEEE, 2005.

\bibitem{hu2011diffusion}
Jinwen Hu, Lihua Xie, and Cishen Zhang.
\newblock Diffusion kalman filtering based on covariance intersection.
\newblock {\em IEEE Transactions on Signal Processing}, 60(2):891--902, 2011.

\bibitem{julier2001general}
S.~J. Julier and J.~K. Uhlmann.
\newblock General decentralized data fusion with covariance intersection
  ({CI}).
\newblock {\em Handbook of Multisensor Data Fusion}, 2001.

\bibitem{julier1997nondivergent}
Simon~J Julier and Jeffrey~K Uhlmann.
\newblock A non-divergent estimation algorithm in the presence of unknown
  correlations.
\newblock In {\em Proceedings of the 1997 American Control Conference},
  volume~4, pages 2369--2373. IEEE, 1997.

\bibitem{julier2007using}
Simon~J Julier and Jeffrey~K Uhlmann.
\newblock Using covariance intersection for slam.
\newblock {\em Robotics and Autonomous Systems}, 55(1):3--20, 2007.

\bibitem{li2013cooperative}
Hao Li and Fawzi Nashashibi.
\newblock Cooperative multi-vehicle localization using split covariance
  intersection filter.
\newblock {\em IEEE Intelligent transportation systems magazine}, 5(2):33--44,
  2013.

\bibitem{lima2021data}
Antoine Lima, Philippe Bonnifait, Veronique Cherfaoui, and Joelle Al~Hage.
\newblock Data fusion with split covariance intersection for cooperative
  perception.
\newblock In {\em 2021 IEEE International Intelligent Transportation Systems
  Conference (ITSC)}, pages 1112--1118. IEEE, 2021.

\bibitem{niehsen2002information}
Wolfgang Niehsen.
\newblock Information fusion based on fast covariance intersection filtering.
\newblock In {\em Proceedings of the Fifth International Conference on
  Information Fusion. FUSION 2002.(IEEE Cat. No. 02EX5997)}, volume~2, pages
  901--904. IEEE, 2002.

\bibitem{noack2017decentralized}
Benjamin Noack, Joris Sijs, Marc Reinhardt, and Uwe~D Hanebeck.
\newblock Decentralized data fusion with inverse covariance intersection.
\newblock {\em Automatica}, 79:35--41, 2017.

\bibitem{petersen2011partitioned}
Arne Petersen and Marc-Andr{\'e} Beyer.
\newblock Partitioned covariance intersection.
\newblock In {\em Proceedings of International Symposium Information on Ships}.
  Citeseer, 2011.

\end{thebibliography}

\end{document}